\documentclass[a4paper,UKenglish,cleveref, autoref, thm-restate]{lipics-v2021}

\pdfoutput=1 
\hideLIPIcs  

\usepackage{amsmath}
\usepackage{amssymb}
\usepackage{amsthm}
\usepackage{graphicx}
\usepackage[colorinlistoftodos]{todonotes}
\usepackage{algorithm}
\usepackage[noend]{algpseudocode}
\usepackage{enumerate}
\usepackage{xcolor}
\usepackage{framed}
\usepackage{bm}
\usepackage{pifont}
\usepackage{multicol}
\usepackage{lipsum}
\usepackage{tikz}
\usetikzlibrary{calc}
\usepackage{wrapfig}
\usepackage{array}
\usepackage{relsize}
\usepackage{comment}
\usepackage{url}
\usepackage[title]{appendix}
\usepackage{subcaption}
\usepackage{xstring}
\usepackage{scalefnt}

\usepackage{booktabs}       

\algdef{SE}[RECEIVING]{Receiving}{EndReceiving}[1]{\textbf{upon
receiving}\ #1\ \algorithmicdo}{\algorithmicend\ \textbf{}}%
\algtext*{EndReceiving}

\algdef{SE}[UPON]{Upon}{EndUpon}[1]{\textbf{upon
}\ #1\ \algorithmicdo}{\algorithmicend\ \textbf{}}%
\algtext*{EndUpon}

\algdef{SE}[DELIVERY]{Delivery}{EndDelivery}[1]{\textbf{upon
delivery}\ #1\ \algorithmicdo}{\algorithmicend\ \textbf{}}%
\algtext*{EndDelivery}

\newcommand{\lr}[1]{\langle #1 \rangle}

\newtheoremstyle{named}{}{}{\itshape}{}{\bfseries}{.}{.5em}{\thmnote{#3 }}
\theoremstyle{named}
\newtheorem*{namedtheorem}{Theorem}

\title{Subquadratic Multivalued Asynchronous Byzantine Agreement WHP} 

\author{Shir Cohen}{Technion, Israel}{shirco@campus.technion.ac.il}{}{Supported by the Adams Fellowship Program of the
Israel Academy of Sciences and Humanities.}
\author{Idit Keidar}{Technion, Israel}{idish@ee.technion.ac.il}{}{}

\authorrunning{S. Cohen and I. Keidar} 

\Copyright{Shir Cohen and Idit Keidar} 

\ccsdesc[500]{Theory of computation~Distributed algorithms}
\ccsdesc[300]{Theory of computation~Cryptographic primitives}
\ccsdesc[500]{Mathematics of computing~Probabilistic algorithms}

\keywords{Byzantine agreement, subquadratic communication, fault tolerance in distributed systems.}

\category{} 

\relatedversion{} 

\nolinenumbers 

\EventEditors{John Q. Open and Joan R. Access}
\EventNoEds{2}
\EventLongTitle{42nd Conference on Very Important Topics (CVIT 2016)}
\EventShortTitle{CVIT 2016}
\EventAcronym{CVIT}
\EventYear{2016}
\EventDate{December 24--27, 2016}
\EventLocation{Little Whinging, United Kingdom}
\EventLogo{}
\SeriesVolume{42}
\ArticleNo{23}

\begin{document}

\maketitle

\begin{abstract}
There have been several reductions from multivalued consensus to binary consensus over the past 20 years. To the best of our knowledge, none of them solved it for Byzantine asynchronous settings. In this paper, we close this gap. Moreover, we do so in subquadratic communication, using newly developed subquadratic binary Byzantine Agreement techniques.\end{abstract}

\section{Introduction}


Byzantine Agreement (BA) is a well-studied problem where a set of correct processes have input values and aim to agree on a common decision despite the presence of malicious ones. 
This problem was first defined over 40 years ago~\cite{lamport2019byzantine}. However, in the past decade BA gained a renewed interest due to the emergence of blockchains as a new distributed and decentralized tool.
Moreover, the scale of the systems in which this problem is solved is much larger than in the past. As a result, there is a constant effort to find new techniques that will enable the reduction of communication complexity of BA solutions.

A significant improvement in BA scalability was enabled by subquadratic solutions, circumventing Dolev and Reischuk’s renown lower bound of $\Omega(n^2)$ messages~\cite{DolevBound}. This was done by King and Saia~\cite{king2011breaking} in the synchronous model and later by Algorand~\cite{Algorand} (first in the synchronous model and then with eventual synchrony). In their work, Algorand presented a validated committee sampling primitive, based on the idea of cryptographic sortition using \emph{verifiable random functions} (VRF)~\cite{micali1999VRF}. This primitive allows different subsets of processes to execute different parts of the BA protocol. Each committee is used for sending exactly one protocol message and messages are sent only by committee members, thus reducing the communication cost.

In this paper, we tackle the asynchronous model, which best describes real-life settings. Importantly, subquadratic asynchronous BA was first introduced not so long ago by Cohen et. al~\cite{cohen2020not} and Blum et. al~\cite{cryptoeprint:2020:851}. 
The limitation of these results is that both solve a binary BA, where the inputs and outputs are in $\{0,1\}$. 
We extend the binary results to multivalued BA, which is more suitable in real-world systems. Nowadays, perhaps the most extensive use of BA solution appears in blockchains to agree on the next block. These blocks carry multiple transactions (as well as additional metadata), which is clearly not a binary value. We note that a similar extension for multivalued consensus was presented in asynchrony by Mostefaoui, Raynal, and Tronel~\cite{mostefaoui2000binary} but it cannot handle Byzantine failures. Another reduction by Turpin and Coan~\cite{turpin1984extending} is able to handle Byzantine failures, but only in a synchronous model. Despite not solving the multivalued case in the Byzantine asynchronous model, both of them have quadratic word complexity (Following the standard complexity notions~\cite{abraham2019asymptotically,mostefaoui2015signature}).

We consider a system with a static set of $n$ processes and an adversary, in the so-called ``permissioned'' setting, where the ids of all processes are well-known. 
The adversary may adaptively corrupt up to  $f=(\frac{1}{3}-\epsilon)n$ processes in the course of a run, where $\frac{1}{2\ln n}<\epsilon < \frac{1}{3}$. In addition, we assume a trusted \emph{public key infrastructure} (PKI) that allows us to use \emph{verifiable random functions} (VRFs)~\cite{micali1999VRF}.
Finally, we assume that once the adversary takes over a process, it cannot “front run” messages that that process had already sent when it was correct, causing the correct messages to be supplanted. This assumption can be replaced if one assumes an erasure model as used in~\cite{cryptoeprint:2020:851,Algorand}. That is, using a separate key to encrypt each message, and deleting the secret key immediately thereafter.

Let us first examine the ``straightforward'' (yet faulty) reduction from multivalued BA to binary BA, both satisfying the \emph{strong unanimity} validity property. This property states that if all correct processes have the same input value, then this must be the decision they all output. To solve the multivalued BA, any process interprets its input as a binary string, and then all processes participate in a sequence of binary BA instances. The input for the $i^{th}$ instance by process $p$ is the $i^{th}$ digit in the binary representation. It is easy to see, that if all correct processes share the same input value, they start each instance with the same binary value and by validity agree upon it. Hence, by the end of the last BA instance, they all reach the same (input) decision. Otherwise, they can agree on some arbitrary common value.

Unfortunately, the simple binary-to-multivalued reduction does not work when applied with existing asynchronous subquadratic solutions. Assume that in the multivalued version of BA, values are taken from a finite domain $\mathcal{V}$. If the size of $\mathcal{V}$ is in $O(n)$, then to represent the input value as a binary string we need $O(\log n)$ bits, and the same number of BA instances. Although this number keeps the overall complexity subquadratic, it breaks the probability arguments made in the existing solutions.
Briefly, both works take advantage of logarithmic subsets of processes that drive the protocol progress. These so-called committees are elected uniformly such that with high probability (WHP) it contains ``enough'' correct processes, and not ``too many'' Byzantine ones. Since, WHP, both algorithms complete in a constant number of rounds, their safety and liveness are guaranteed WHP. However, once we apply these techniques and make the probability arguments more than a constant number of times (as we would if we were to apply the reduction), the high probability does not remain high at all.

To overcome this challenge, we take a different approach. We generalize the method in~\cite{turpin1984extending} to work with asynchronous committee sampling and solve for \emph{weak unanimity} validity. Our algorithm requires only two additional committees, compared to any binary BA algorithm.
Finally, we present the first multivalued BA with a word complexity of $\widetilde{O}(n)$.

\section{Model and Preliminaries}
\label{sec:model}

We consider a distributed system consisting of a well-known static set $\Pi$ of $n$ processes and an adversary.
The adversary may adaptively corrupt up to  $f=(\frac{1}{3}-\epsilon)n$ processes in the course of a run, where $\frac{1}{2\ln n}<\epsilon < \frac{1}{3}$.
A corrupted process is \emph{Byzantine}; it may deviate arbitrarily from the protocol. In particular, it may crash, fail to send or receive messages, and send arbitrary messages. As long as a process is not corrupted by the adversary, it is \emph{correct} and follows the protocol.
In addition, we assume that once the adversary takes over a process, it cannot “front run” messages that that process had already sent when it was correct, causing the correct messages to be supplanted. This assumption can be replaced if one assumes an erasure model as used in~\cite{cryptoeprint:2020:851,Algorand}. That is, using a separate key to encrypt each message, and deleting the secret key immediately thereafter.



{\bf Cryptographic tools.}
We assume a trusted PKI, where private and public keys for the processes are generated before the protocol begins and processes cannot manipulate their public keys. We denote a value $v$ signed by process $p_i$ as $\lr{v}_i$. In addition, we assume that the adversary is computationally bounded, meaning that it cannot obtain the private keys of processes unless it corrupts them. Furthermore, we assume that the PKI is in place from the outset. 

These assumptions allow us to use a verifiable random function (VRF). This is a pseudorandom function that provides proof of its correct computation~\cite{micali1999VRF}. In our context, the use of VRF is implicit and it is used to implement a \emph{validated committee sampling} as defined in Section~\ref{sec:committees}.


{\bf Communication.}
We assume that every pair of processes is connected via a reliable link. Messages are authenticated in the sense that if a correct process $p_i$ receives a message $m$ indicating that $m $ was sent by a correct process $p_j$, then $m$ was indeed generated by $p_j$ and sent to $p_i$. The network is asynchronous, i.e., there is no bound on message delays.

{\bf Complexity.}
We use the following standard complexity notions~\cite{abraham2019asymptotically,mostefaoui2015signature}.
While measuring complexity, we allow a \emph{word} to contain a signature or a value from a finite domain.
We measure the expected \emph{word complexity} of our protocols as
the maximum of the expected total number of words sent by correct processes where the maximum is computed over all inputs and applicable adversaries and expectation is taken over the random variables.



{\bf Validated Committee Sampling.}
Using VRFs, it is possible to implement \emph{validated committee sampling}, which is a primitive that allows processes to elect committees without communication and later prove
their election. It provides every process $p_i$ with a private
function $\emph{sample}_i(s,\lambda)$, which gets a string $s$ and a
threshold $1\leq \lambda \leq n$ and returns a tuple $\lr{v_i,\sigma_i}$, where
$v_i \in \{ \emph{true},\emph{false}\}$ and $\sigma_i$ is a proof that $v_i=\emph{sample}_i(s,\lambda)$. If $v_i=\emph{true}$ we say that $p_i$ is \emph{sampled} to the committee for $s$ and $\lambda$. The primitive ensures that $p_i$ is sampled with probability $\frac{\lambda}{n}$.
In addition, there is a public (known to all) function,
$\emph{committee-val}(s,\lambda,i,\sigma_i)$, which gets a
string $s$, a threshold $\lambda$, a process identification $i$ and a proof
$\sigma_i$, and returns \emph{true} or \emph{false}.

Consider a string $s$.
For every $i$, $1 \leq i \leq n$, let $\lr{v_i,\sigma_i}$ be the return value of $\emph{sample}_i(s,\lambda)$. The following is satisfied for every $p_i$:
\begin{itemize}
  
  \item $\emph{committee-val}(s,\lambda,i,\sigma_i) = v_i$.
  
  \item If $p_i$ is correct, then it is infeasible
  for the adversary to compute $\emph{sample}_i(s,\lambda)$.
  
  \item It is infeasible for the adversary to find  $\lr{v, \sigma}$
  s.t.\ $v \neq v_i$ and
  $\emph{committee-val}(s,\lambda,i,\sigma) = true$.
  
\end{itemize}


Due to space limitations, we present here the parameters and guarantees as presented and proven in~\cite{cohen2020not} using Chernoff bounds. For simplicity, we only state the claims we are using in this paper.

\begin{namedtheorem}[Committee Sampling Properties from~\cite{cohen2020not}]


Let the set of processes sampled to the committee for $s$ and $\lambda$ be $C(s,\lambda)$, where $\lambda$ is set to $8\ln n$.
Let $d$ be a parameter of the system such that $\frac{1}{\lambda}<d < \frac{\epsilon}{3}-\frac{1}{3\lambda}$.
We set $W\triangleq \left \lceil{(\frac{2}{3}+3d)\lambda}\right \rceil$ and $B\triangleq \left \lfloor{(\frac{1}{3}-d)\lambda}\right \rfloor$.
With high probability the following hold:
\begin{description}
  \item[(S3)] At least
  $W$ processes in $C(s,\lambda)$ are correct.
    \item[(S4)] At most $B$ processes in $C(s,\lambda)$ are Byzantine.

    \item[(S5)] Consider $C(s,\lambda)$ for some string $s$ and two sets $P_1,P_2\subset C(s,\lambda)$ s.t $|P_1|=|P_2|=W$. Then, $|P_1\cap P_2|\geq B +1$.
\end{description}

\end{namedtheorem}




\section{From Binary BA to Multivalued BA}

In the Byzantine Agreement (BA) problem, a set $\Pi$ of $n$ processes attempt to reach a common decision. In addition, the decided value must be ``valid'' in some sense which makes the problem non-trivial. 
We consider two standard variants of BA for asynchrony that differ in their validity condition and the domain of inputs by processes in the system.
In the binary version, all inputs are taken from the domain $\{0,1\}$, while in the multivalued case they can be any value from any finite domain $\mathcal{V}$.
For the validity condition, in order to support a larger domain we weaken the validity condition. Instead of the known \emph{strong unanimity} property, we opt for \emph{weak unanimity} as defined below.
In this work we show how to reduce a subquadratic weak multivalued BA to a subquadratic binary strong BA, both are solved WHP.
That is, a probability that tends to 1 as $n$ goes to infinity.
Formally, we take a black-box solution to:

\begin{definition}[Binary Strong Byzantine Agreement WHP]
In Binary Strong Byzantine Agreement WHP, each correct process $p_i\in\Pi$ proposes a binary input value $v_i$ and decides on an output value $decision_i$ s.t. with high probability the following properties hold:

\begin{itemize}
    \item Validity (Strong Unanimity). If all correct processes propose the same value $v$, then any correct process
that decides, decides $v$.
    \item Agreement. No two correct processes decide differently.
    \item Termination. Every correct process eventually decides.
\end{itemize}

\end{definition}

And use it to solve:


\begin{definition}[Multivalued Weak Byzantine Agreement WHP]
In Multivalued Weak Byzantine Agreement WHP, each correct process $p_i\in\Pi$ proposes an input value $v_i$ and decides on an output value $decision_i$ s.t. with high probability the following properties hold:

\begin{itemize}
    \item Validity (Weak Unanimity). If all processes are correct and propose the same value $v$, then any correct process that decides, decides $v$.
    \item Agreement. Same as above.
    \item Termination. Same as above.
\end{itemize}

\end{definition}

\begin{algorithm}[t] \scalefont{0.9}
\caption{Multivalued Byzantine Agreement($v_i$): code for $p_i$}


\begin{algorithmic}[1]

\Statex local variables: $\emph{alert}\in \{true,false\}$, initially $false$
\Statex $\emph{count}\in \mathbb{N}$, initially $0$
\Statex $\emph{init-set},\emph{init-values-set},\emph{converge-set}\in \mathcal{P}(\Pi)$, initially $0$
%

\If{$\emph{sample}_i(\textsc{init},\lambda) = $ \emph{true}}
	 broadcast $\lr{\textsc{init},v_i}_i$ \label{l:send_init}
\EndIf

\Receiving{ $\lr{\textsc{init},v_j}_j$ with valid $v_j$ from validly sampled $p_j$ }

	\State \emph{init-set} $\gets$ \emph{init-set} $\cup  \{j\}$
	\State \emph{init-values-set} $\gets$ \emph{init-values-set} $\cup  \{v_j\}$   

    \If{$\emph{sample}_i(\textsc{converge},\lambda) = $ \emph{true} and $|\emph{init-set}|=W$ for the first time} \label{l:terminate_init}

    	\If{$\emph{init-values-set}=\{v_i\}$} \Comment{All received values are $p_i$'s initial value} \label{l:signleton}
        

        \State batch the $W$ messages into $QC_{v_i}$ \label{l:batch_qc}

    	\State send $\lr{\textsc{converge},\emph{true}, QC_{v_i}}_i$  to all processes \label{l:send_qc}

    	\Else
    	
    	\State send $\lr{\textsc{converge},\emph{false},\bot}_i$  to all processes \label{l:not_content}

        \EndIf
    \EndIf
\EndReceiving

\Receiving{ $\lr{\textsc{converge},\emph{is\_content},QC_{v}}_j$  from validly sampled $p_j$ }

	\State \emph{converge-set} $\gets$ \emph{converge-set} $\cup  \{j\}$

    \If{\emph{is\_content}=\emph{true}} \label{l:if_content}
    \State $\emph{count} +=1$ \label{l:inc_counter}

    \EndIf

	\State \textbf{when} $|\emph{converge-set}|=W$ for the first time  \label{l:terminate_converge}

    \State \hspace{\algorithmicindent} $\emph{alert}\gets \emph{count}<B+1$ \label{l:check_count}
    \State \hspace{\algorithmicindent} $\emph{binary\_decision}_i\gets \emph{Binary Byzantine Agreement(alert)}$ \label{l:binary_dec}

    \State\hspace{\algorithmicindent} \textbf{if} $\emph{binary\_decision}_i=true$ \textbf{then}
    
    \State\hspace{\algorithmicindent} \hspace{\algorithmicindent}$\emph{decision}_i\gets \bot$ \label{l:dec_def}
    
    \hspace{\algorithmicindent}\textbf{else}
    
        



        \State\hspace{\algorithmicindent}\hspace{\algorithmicindent} wait for a message of the form $\lr{\textsc{converge},\emph{true},QC_{v}}_j$ from validly sampled $p_j$ if such was not already received

        \State\hspace{\algorithmicindent}\hspace{\algorithmicindent}$\emph{decision}_i\gets v$ \label{l:dec_by_qc}


\EndReceiving

\end{algorithmic}

\label{alg:to_multi}
\end{algorithm}

We employ committee sampling and a binary subquadratic strong BA to present a multivalued solution to the weak BA problem. 
That is, the processes' initial values are from an arbitrary domain $\mathcal{V}$.
We follow the method presented in~\cite{turpin1984extending} and adjust it to work with an asynchronous environment and committee sampling to achieve a subquadratic solution.
The algorithm, presented in Algorithm~\ref{alg:to_multi}, consists of two communication phases followed by a binary BA execution. To reduce the communication costs of the algorithm, the two phases are being executed only by a subset of the processes that are elected uniformly in random by a committee sampling primitive. 

The first step is an \textsc{init} step, in which all \textsc{init} committee members send their signed initial value to all other processes (line~\ref{l:send_init}). The second is a \textsc{converge} step, during which all \textsc{converge} committee members aim to converge around one common value. To do so, \textsc{converge} processes are waiting to hear from sufficiently many processes in the \textsc{init} committee.
Since committees are elected using randomization, it is impossible for processes to wait for all of the previous committee members, as the size of the committee is unknown. Instead, it is guaranteed that a process hears from at least $W$ processes WHP.

For some process $p$ in the \textsc{converge} committee, if all $W$ \textsc{init} messages include the same value $v$, that is also $p$'s initial value, then $p$ is considered to be \emph{content}. To inform all other processes, $p$ sends a \textsc{converge} message claiming to be content, that also carries a quorum certificate (QC) on the value $v$ containing all received messages (line~\ref{l:send_qc}). In the complementary case, where $p$ knows of at least two different values by line~\ref{l:signleton}, it sends a \textsc{converge} message with a false is\_content flag (line~\ref{l:not_content}).

In the third part of the algorithm, processes run a binary consensus whose target is deciding whether the system as a whole is content. To do so, processes update an alert flag that is determined according to the number of content processes in the \textsc{converge} committee (lines ~\ref{l:if_content} -- \ref{l:check_count}). It is designed such that if a correct process $p$ hears from at least one non-content correct process, it sets its alert flag to true. To do so, we utilize the parameter $B$ which is, according to the specification, an upper bound on the number of Byzantine processes in a committee.

After processes set their boolean \emph{alert} flag, they make a binary decision on its values at line~\ref{l:binary_dec}. If the output is \emph{true}, then all correct processes output $\bot$ (line~\ref{l:dec_def}). Notice that by the strong unanimity property of the binary BA, it is impossible that all correct processes have had \emph{alert}=\emph{false} before the execution. Namely, there were many non-content processes in the converge committee. Otherwise, if the binary decision is \emph{true}, then it must be the case that (i) all content processes are content with respect to the same value and (ii) at least one correct process was content in the converge committee  (proof appears in the appendix). In this case, that process carries a quorum certificate with $W$ different signatures on a value $v$, that can be safely decided upon (line~\ref{l:dec_by_qc}) and is guaranteed to eventually arrive at all correct processes. This is mainly thanks to the fact that two subsets of the converge committee of size $W$ intersect by at least one correct process.  
In the appendix, we prove the following theorem:

\begin{theorem}
Algorithm~\ref{alg:to_multi} implements multivalued weak Byzantine Agreement WHP with a word complexity of $\widetilde{O}(n)$.
\end{theorem}

\section{Conclusions}
\label{sec:conclusions}

Real-world systems are asynchronous and prone to Byzantine failures. This paper presents an algorithm that reduces the multivalued weak BA WHP to binary strong BA. Together with the binary BA presented in~\cite{cohen2020not,cryptoeprint:2020:851} this paper yields that first subquadratic multivalued BA.

\bibliography{references}

\newpage

\appendix

\section{Validated Committee Sampling} 
\label{sec:committees}

A common strategy for achieving subquadratic word complexity is to avoid all-to-all communications~\cite{Algorand,king2011breaking}. This is done by sampling subsets of processes such that only those elected to a committee send messages. Due to committees being randomly sampled, the adversary is prevented from corrupting their members. On the other hand, committee members cannot predict the following committee sample and thus send their messages to all processes. Potentially, if the committee is sufficiently small, this technique allows committee-based protocols to result in subquadratic word complexity.

Using VRFs, it is possible to implement \emph{validated committee sampling}, which is a primitive that allows processes to elect committees without communication and later prove
their election. It provides every process $p_i$ with a private
function $\emph{sample}_i(s,\lambda)$, which gets a string $s$ and a
threshold $1\leq \lambda \leq n$ and returns a tuple $\lr{v_i,\sigma_i}$, where
$v_i \in \{ \emph{true},\emph{false}\}$ and $\sigma_i$ is a proof that $v_i=\emph{sample}_i(s,\lambda)$. If $v_i=\emph{true}$ we say that $p_i$ is \emph{sampled} to the committee for $s$ and $\lambda$. The primitive ensures that $p_i$ is sampled with probability $\frac{\lambda}{n}$.
In addition, there is a public (known to all) function,
$\emph{committee-val}(s,\lambda,i,\sigma_i)$, which gets a
string $s$, a threshold $\lambda$, a process identification $i$ and a proof
$\sigma_i$, and returns \emph{true} or \emph{false}.

Consider a string $s$.
For every $i$, $1 \leq i \leq n$, let $\lr{v_i,\sigma_i}$ be the return value of $\emph{sample}_i(s,\lambda)$. The following is satisfied for every $p_i$:

\begin{itemize}
  
  \item $\emph{committee-val}(s,\lambda,i,\sigma_i) = v_i$.
  
  \item If $p_i$ is correct, then it is infeasible
  for the adversary to compute $\emph{sample}_i(s,\lambda)$.
  
  \item It is infeasible for the adversary to find  $\lr{v, \sigma}$
  s.t.\ $v \neq v_i$ and
  $\emph{committee-val}(s,\lambda,i,\sigma) = true$.
  
\end{itemize}


In this work, we take advantage of the validated committee sampling, as defined and used in~\cite{cohen2020not}. In their work, Cohen et. al adjusted the use of this primitive to the asynchronous setting. With their technique, some asynchronous protocols can be modified to use committees. The challenge in asynchrony, with the lack of timeouts, is that no process can wait to hear from an entire committee nor can it wait for $n-f$ processes as usually done in such algorithms. Instead, some parameters are defined with respect to the system parameters such that safety and liveness can be provided with high probability. Due to space limitations, we present here the parameters and guarantees as presented and proven in~\cite{cohen2020not} using Chernoff bounds. For simplicity, we only state the claims we are using in this paper.

\begin{theorem}

Let the set of processes sampled to the committee for $s$ and $\lambda$ be $C(s,\lambda)$, where $\lambda$ is set to $8\ln n$.
Let $d$ be a parameter of the system such that $\frac{1}{\lambda}<d < \frac{\epsilon}{3}-\frac{1}{3\lambda}$.
We set $W\triangleq \left \lceil{(\frac{2}{3}+3d)\lambda}\right \rceil$ and $B\triangleq \left \lfloor{(\frac{1}{3}-d)\lambda}\right \rfloor$.
With high probability the following hold:
\begin{description}
  \item[(S3)] At least
  $W$ processes in $C(s,\lambda)$ are correct.
    \item[(S4)] At most $B$ processes in $C(s,\lambda)$ are Byzantine.

    \item[(S5)] Consider $C(s,\lambda)$ for some string $s$ and two sets $P_1,P_2\subset C(s,\lambda)$ s.t $|P_1|=|P_2|=W$. Then, $|P_1\cap P_2|\geq B +1$.
\end{description}

\end{theorem}

\section{Correctness Proof of Algorithm~\ref{alg:to_multi}}
\label{subsec:proof}

\begin{lemma}{(Validity)}
If all processes are correct and they all propose the same value $v$, then they all decide $v$ WHP.
\end{lemma}
\begin{proof}

If all processes are correct and propose the same value $v$, then all members of the \textsc{init} committee broadcast $v$. By S3, WHP, there are at least $W$ correct processes in the \textsc{init} committee. Hence, every process in the \textsc{converge} committee receives $W$ \textsc{init} messages with the value $v$ and by assumption $v$ is its initial value. Thus, WHP all correct processes in the \textsc{converge} committee execute lines~\ref{l:batch_qc} -- \ref{l:send_qc} and send a \textsc{converge} message with a quorum certificate on $v$.

By S3, WHP, there are at least $W$ correct processes in the \textsc{converge} committee.
Hence, all correct processes receive $W$ \textsc{converge} messages with a $\emph{is\_content}=true$ and a quorum certificate on $v$, for each of them they increase the counter at line~\ref{l:inc_counter}.
Following that, and since $W>B$, all correct processes set alert to \emph{false} at line~\ref{l:check_count} and start the binary BA with the same input \emph{false}.
By the strong unanimity validity condition of the binary BA, they all decide on \emph{false} at line~\ref{l:binary_dec}.
Since they all previously received a \textsc{converge} messages with a quorum certificate on $v$, they all decide on $v$ at line~\ref{l:dec_by_qc}.


\end{proof}

\begin{lemma}
\label{lemma:evn_qc}
If the binary decision is \emph{false} at line~\ref{l:binary_dec}, then all correct processes receive a message $\lr{\textsc{converge},\emph{true},QC_{v}}_j$ from validly sampled $p_j$ WHP.
\end{lemma}
\begin{proof}
Notice that there exists at least one correct process $p_i$ that initiates the binary agreement with $alert=false$, as otherwise by the strong unanimity validity condition of the binary BA the binary decision must have been true. It follows from the code (line~\ref{l:inc_counter}, \ref{l:check_count}) that $p_i$ received at least $B+1$ \textsc{converge} messages with \emph{is\_content=true}. By property S4, WHP, there are at most $B$ Byzantine processes in the \textsc{converge} committee. Hence, $p_i$ received the \textsc{converge} message with a QC from at least one correct process $p_j$. By the code, $p_j$ must have sent that message to all correct processes at line~\ref{l:send_qc}. Therefore, every correct process eventually receives that message.

\end{proof}

\begin{lemma}{(Termination)}
Every correct process decides WHP.

\end{lemma}
\begin{proof}
By S3, WHP, $W$ correct processes are elected to the \textsc{init} committee, broadcasting an \textsc{init} message. Thus, every correct process who is sampled to the \textsc{converge} committee eventually receives W \textsc{init} messages WHP and enters the if condition at line~\ref{l:terminate_init}. By S3 again, WHP, at least $W$ correct processes are elected to the \textsc{converge} committee. Thus, every correct process receives WHP $W$ \textsc{converge} messages, and execute line~\ref{l:terminate_converge}. By the termination of the binary BA, all correct processes terminate the binary BA WHP.
If the binary decision is \emph{true}, then they all decide at line~\ref{l:dec_def}. Otherwise, by Lemma~\ref{lemma:evn_qc} they all decide at line~\ref{l:dec_by_qc} after receiving a valid $\lr{\textsc{converge},\emph{true},QC_{v}}_j$ message WHP.
\end{proof}

\begin{lemma}{(Agreement)}
No two correct processes decide different values WHP.
\end{lemma}
\begin{proof}
According to the binary BA agreement, all correct processes decide the same value WHP at line~\ref{l:binary_dec}.
We consider the two possible binary BA outputs.
If the binary decision is \emph{true}, then all correct processes decide the domain's default value, and hence the same value. 
Otherwise, the binary decision is \emph{false} and the decision is made according to the value signed in a QC received in a converge message. By Lemma~\ref{lemma:evn_qc}, WHP all correct processes receive a converge message with \emph{is\_content} flag set to \emph{true}. It is left to show that all of these messages carry a quorum certificate for the same value $v$.

A valid QC on value $v$ is a collection of $W$ signatures by different processes that sent an \textsc{init} message on the value $v$. By property S5, every two subsets of the \textsc{init} committee of size $W$ intersect by at least one correct process. Since a correct process only signs one \textsc{init} message it follows that only one valid QC can be sent at line~\ref{l:send_qc}.

\end{proof}

\textbf{Complexity.}
In Algorithm~\ref{alg:to_multi} all correct processes that are sampled to the two committees (lines~\ref{l:send_init},\ref{l:terminate_init}) send messages to all other processes. Each of these messages contains a value from the finite domain, a VRF proof of the sender’s election to the committee, and possibly a quorum certificate of $W$ different signatures.
Therefore, each message’s size is either a constant number of words or $W$ words. Thus, the total word complexity of a multivalued weak BA WHP is $O(nWC)$ where $C$ is the number of processes that are sampled to the committees. Since each process is sampled to a committee with probability $\frac{\lambda}{n}$, we get a word complexity of $O(nW\lambda) = O(n\log ^2 n) = \widetilde{O}(n)$ in expectation.

We conclude:

\begin{theorem}
Algorithm~\ref{alg:to_multi} implements multivalued weak Byzantine Agreement WHP with a word complexity of $\widetilde{O}(n)$.
\end{theorem}

\end{document}